\documentclass[12pt]{article}
\pdfoutput=1
\usepackage{amsmath}
\usepackage{graphicx,psfrag,epsf}
\usepackage{enumerate}
\usepackage{natbib}
\usepackage{url} 

\usepackage{amssymb, amsfonts, amscd, xspace, pifont, amsthm}
\usepackage{mathrsfs}
\usepackage{booktabs}
\usepackage{multirow}
\usepackage{epstopdf}
\usepackage{color}
\usepackage{ulem}
\usepackage{tikz,verbatim}
\usepackage{caption,subcaption,multirow}
\usepackage[nottoc,numbib]{tocbibind}
\usepackage{float}

\theoremstyle{plain}
\newtheorem{theorem}{Theorem}[section]
\newtheorem{corollary}{Corollary}[theorem]
\newcommand{\blind}{0}

\newcommand{\V}[1]{\ensuremath{\boldsymbol{#1}}\xspace}
\newcommand{\M}[1]{\ensuremath{\boldsymbol{#1}}\xspace}

\newtheorem{lemma}{Lemma}[section]

\numberwithin{equation}{section}
\def\twoImages#1#2#3#4#5#6 
{
  \centerline{\hfill\makebox[#2]{#3}\hfill\makebox[#5]{#6}\hfill}
  \centerline{\hfill
    \includegraphics[width=#2]{#1}
    \hfill
    \includegraphics[width=#5]{#4}
    \hfill}
}

\newenvironment{example}[1][Example]{\begin{trivlist}
\item[\hskip \labelsep {\bfseries #1}]}{\end{trivlist}}

\begin{document}

\def\spacingset#1{\renewcommand{\baselinestretch}%
{#1}\small\normalsize} \spacingset{1}

\if0\blind
{
  \title{\bf Generalized linear models with low rank effects for network data}
  \author{
    Yun-Jhong Wu, Elizaveta Levina, and Ji Zhu  \\
    Department of Statistics, University of Michigan}
  \maketitle
} \fi

\if1\blind
{
  \bigskip
  \bigskip
  \bigskip
  \begin{center}
    {\LARGE\bf Link prediction for partially observed networks}
\end{center}
  \medskip
} \fi

\bigskip
\begin{abstract}
Networks are a useful representation for data on connections between units of interests, but the observed connections are often noisy and/or include missing values.  One common approach to network analysis is to treat the network as a realization from a random graph model, and estimate the underlying edge probability matrix, which is sometimes referred to as network denoising.  Here we propose a generalized linear model with low rank effects to model network edges.  This model can be applied to various types of networks, including directed and undirected, binary and weighted, and it can naturally utilize additional information such as node and/or edge covariates. We develop an efficient projected gradient ascent algorithm to fit the model, establish asymptotic consistency, and demonstrate empirical performance of the method on both simulated and real networks.
\end{abstract}

\noindent%
\vfill

\newpage
\spacingset{1.45} 

\section{Introduction}
Networks are widely used to represent and analyze data in many domains, for example, for social, biological, and communication systems.  Each network consists of nodes and edges.  For example, in social networks, nodes may correspond to people and edges represent friendships; in biological networks, nodes may correspond to genes or proteins while edges represent regulatory relationships.  Besides nodes and edges, other information is often available in the form of node and/or edge covariates, such as people's demographic information or the closeness of a friendship  in social networks, and proteins' chemical components or the strength of the regulatory relationship in biological networks.  

One fundamental problem in network analysis is to understand the mechanism that generates the edges by estimating the expectation of the adjacency matrix, sometimes referred to as network denoising.  The expectation gives probabilities of links for every pair, which can be further used to perform link prediction;  in fact for link prediction any monotone transformation of the link probabilities is sufficient.    For binary networks, link prediction can be framed as a classification problem, which presence/absence of edge as the class label for each paper, and some sort of score for each pair of nodes (e.g. an estimated probability of link) used to predict the class.  

Most approaches to the estimating the probabilities of edges (or more generally scores) use the information from node features when available, and/or network topology such as the number of common neighbors, etc.    Many approaches are based on homophily, which means that the more ``similar'' two nodes are, the more likely they are to become connected.  Homophily has been widely observed in social networks \citep{McPherson2001} and other contexts \citep{Zhou2009}.   If homophily is assumed, estimating adjacency matrices is closely related to the question of how to measure similarity between nodes. For node features, any appropriate similarity measure for vectors can be used.    Multiple measures based on network topology are also available;  see e.g., Section 3 in \cite{Lu2011}.   Other proposals include aggregating several similarities such as the number of 2-path and 3-path between two nodes \citep{Zhou2009} and using kernels to measure the similarity between node pairs and combining it with the support vector machine (SVM) for classification in the context of estimating protein-protein interactions \citep{Ben-Hur2005}. 

Alternatively, one can embed nodes in an Euclidean space and measure the similarity between nodes according to the distance between the nodes' latent positions. This approach includes various probabilistic network models such as the latent space model \citep{Hoff2002}, the latent variable model \citep{Hoff2007}, the latent feature model \citep{Miller2009}, the latent factor model \citep{Hoff2008}, the latent variable models with Gaussian mixture positions \cite{Krivitskyetal2009}, and the Dirichlet network model\cite{Williamson2016}.   In all these models, the latent positions have to be estimated via Markov Chain Monte Carlo (MCMC), which is very time consuming.  More computationally efficient approaches have been developed.  For example, the leading eigenvectors  of the graph Laplacian can be used to embed the nodes in a low-dimensional space \citep[e.g.][]{Kunegis2009} by spectral decomposition, and their embedding coordinates can be veiwed at the latent node positions. Other recent efforts have been devoted to fitting latent space models by stochastic variational inference \citep{Zhu2012a}
and gradient descent algorithms \citep{Ma2017}.    The latter paper was written simultaneously and independently of the current work, and while it uses a similar algorithm in optimizaiton, it fits a different model, focuses on the problem of latent position estimation rather than link prediction, and, unlike ours, does not cover the directed case.   

In another related line of work, graphon estimation methods estimate the edge probability matrix under node exchangeability and various additional assumptions on the matrix (smoothness, low-rankness, etc)  \citep[e.g.][]{Choi2014,Yang2014,Olhede2014,Gao2015,Zhang2015}. However, when node or edge features are available, exchangeability does not apply.    Instead, a common approach is to aggregate information on the features and  multiple similarity indexes to create a single score for predicting links. For example, \cite{Kashima2009} and \cite{Menon2011} treat topology-based similarities as edge attributes and propose an SVM-based approach for edge estimation.     

Assumptions other than homophily have also been considered, such as hierarchical network structure \citep{Clauset2008}, structural equivalence \citep{Hoff2007}.  In another approach, \citet{Zhao2013} used pair similarity instead of node similarity for edge prediction, arguing that edges between similar pairs of nodes should have similar probability of occurring.

The problem of link prediction is also related to the problem of matrix completion, which is commonly solved under low rank constraints \citep[e.g.][]{Candes2009}.   In fact if the network is undirected and binary without any covariates, our proposed method is equivalent to the 1-bit matrix completion algorithm of  \cite{Davenport2013}, who established consistency of the maximum likelihood estimator for this setting. However, the 1-bit matrix completion formulation is much narrower:  it does not allow for covariates and, crucially, assumes that the links are missing completely at random with equal probability, which is not a realistic assumption for networks. 

The model we propose here represents the probability of an edge through a small number of parameters, like the latent space models; but unlike previous work, all we assume is a general low rank structure, without requiring anything more specific.   This makes our method easily applicable to many types of networks: directed and undirected, binary and weighted, with and without node/edge covariates.    Unlike latent space models, we do not require computationally expensive MCMC;  instead, we fit the proposed model through an efficient projected gradient algorithm.  In addition to computational efficiency, our method also has attractive theoretical properties, such as consistency for network estimation under the low rank assumption on the true probability matrix.  

The rest of this article is organized as follows.   The proposed model and the estimation algorithm are presented in Section \ref{sec_method}. In Section \ref{sec_consistency}, we establish several theoretical guarantees including consistency.  Numerical evaluation of the proposed method and comparisons to other network estimation approaches on simulated networks are presented in Section \ref{sec_simulation}.  In Section \ref{sec_data_analysis}, we illustrate the proposed method on two real networks, the friendship network of users of the Last.fm music website and the C. Elegans neural network.  Section \ref{sec_discussion} concludes the paper with discussion and future work.  All proofs are given in the Appendix.

  S\section{Generalized linear models for network data with low rank effects}\label{sec_method}
We start with setting up notation.   The data consist of a single observed $n\times n$ adjacency matrix $\M{A}=[A_{ij}]_{n\times n}$, where $A_{ij}$ represents the edge from node $i$ to node $j$, which can be either a value binary (0/1) or a weight.  If additional information on nodes and/or edges is available, we represent it as an $m$-dimensional attribute vector for each pair of nodes $i$ and $j$, denoted by $\boldsymbol{x}_{ij}=(x_{ij1},\dots,x_{ijm})^\top$.   If the attributes are attached to nodes rather than edges, we convert them to edge attributes using a similarity measure, discussed in more detail below.   Our goal is to compute a score for each pair of nodes to represent the strength of an edge that may connect them. A natural score is the expected value $\M{P}=[p_{ij}]_{n\times n}=\mathbb{E}[\M{A}]$.   Then we can view the problem as a generalized regression question, fitting the model 
\[
p_{ij}=g(\boldsymbol{x}_{ij}),
\]
where $g$ is a mean function. 

\subsection{Generalized linear models for network data}
A natural way to connect covariates to the strength of network edges is to use the generalized linear model (GLM).   For example, logistic regression and logit models have been used for fitting binary directed networks \citep{Wasserman1996}. It is straightforward to generalize this approach to various types of networks by considering a generalized linear model
\begin{align}
L(p_{ij})=\theta_{ij} + \boldsymbol{x}_{ij}^\top \V{\beta}, \label{lrglm}
\end{align}
where $L$ is a link function to be specified and $\V{\beta}\in\mathbb{R}^{m}$ is a vector of coefficients.  As normally done in GLM, we assume that the distribution of $A_{ij}$ only depends on covariates through their linear combination with an unknown coefficient vector $\V{\beta}$,  
and that edges are independent conditional on covariates. The parameter $\theta_{ij}$ represents an interaction between nodes $i$ and $j$ for $i,j=1,\dots,n$.    Further assuming an exponential family distribution,  the conditional distribution of $\M{A}$ with the mean matrix $\M{P}$ takes the canonical form
\begin{align}
f_{(\M{\Theta},\V{\beta})} (\M{A} \mid \mathcal{X})&=\prod_{ij}f_{(\theta_{ij},\V{\beta})}(A_{ij}\mid \boldsymbol{x}_{ij}) =\prod_{ij}c(A_{ij})\exp\bigg(\eta_{ij} A_{ij} -b(\eta_{ij})\bigg), \label{lrem:canonical_form}
\end{align}
where $\M{\Theta}=[\theta_{ij}]_{[n \times n]}$, $\mathcal{X}=[\M{X}_1,\dots,\M{X}_m]\in\mathbb{R}^{n\times n\times m}$,  $\M{X}_k=[x_{ijk}]_{n\times n}$, $k=1, \ldots, m$, $\eta_{ij}=\theta_{ij}+\boldsymbol{x}_{ij}^\top \V{\beta}$, and the corresponding canonical link function is given by $L^{-1}=g=b'$. 
This general setting includes, for example, the logistic model for fitting binary networks and binomial and Poisson models for integer-weighted networks. Extending it to multinomial logistic models for networks with signed or labeled edges is also straightforward.  

Model \eqref{lrglm} involves more parameters than can be fitted without regularization or additional assumptions on $\M{\Theta}$. One possibility is to impose regularization through the commonly occurring dependency among edges in networks known as transitivity: if A and B are friends, and B and C are friends, then A and C are more likely to be friends.   This idea has been utilized by \cite{Hoff2005}, in which the random effects model was extended to the so-called bilinear mixed-effects model to model the joint distribution of adjacent edges.  Here we take a different and perhaps more general approach by imposing a low rank constraint on the effects matrix, implicitly inducing sharing information among the edges; this allows us to both model individual node effects and share information, which seems to be more appropriate for network data than the random effects modeling assumption of random and identically distributed $\theta_{ij}$'s.   

\subsection{The low rank effects model}
In general, regularization can be applied to either $\M{\Theta}$ or $\V{\beta}$ or both;  a sparsity constraint on $\V{\beta}$ would be natural when the number of attributes $m$ is large, but the more important parameter to constrain here is $\M{\Theta}$, which contains $n^2$ parameters.   A natural general constraint that imposes structure without parametric assumptions is constraining the rank of $\M{\Theta}$, assuming 
\begin{align}
L(\M{P})=\M{\Theta}+\mathcal{X}\otimes\V{\beta}, ~~\mathrm{rank} (\M{\Theta}) \leq r, \label{lrglm_mat}
\end{align}
where $\mathcal{X}\otimes\V{\beta}=\sum^m_{k=1}\beta_k\mathbf{X}_k$, and, in a slight abuse of notation, $L(\M{P})$ is the link function applied element-wise to the matrix $\M{P}$.  

The rank constrained model \eqref{lrglm_mat} is related to latent space models, for example, the eigenmodel proposed by \cite{Hoff2007} for undirected binary networks.  The projection model assumes that the edge probability is given by
\begin{align}
\mathrm{logit}(p_{ij})=\alpha + \mathbf{z}_i^\top \mathbf{\Lambda}\mathbf{z}_j+\boldsymbol{x}_{ij}^\top\V{\beta}, \label{proj}
\end{align}
where $\mathbf{z}_i \in \mathbb{R}^{(r-1)}$ represents the position of node $i$ in a latent space. Note that the $n\times n$ matrix $\alpha\mathbf{11}^\top+\mathbf{Z}\mathbf{\Lambda}\mathbf{Z}^\top$, where $\mathbf{Z}=[\mathbf{z}_1 \cdots \mathbf{z}_n]^\top\in\mathbb{R}^{n\times {(r-1)}}$, is at most of rank $r$. By setting $L$ to be the logit link, the eigenmodel can be obtained as a special case of the  low rank effects model \eqref{lrglm_mat}, although the fitting method proposed for the eigenmodel by \cite{Hoff2007} is much more computationally intensive.   

Full identifiability for \eqref{lrglm_mat} requires additional assumptions, even though the mean matrix $\M{P}$ is always identifiable and so is $\mathbf{\Theta}+\mathcal{X}\otimes\V{\beta}$.   For $\mathbf{\Theta}$ and $\V{\beta}$ to be individually identifiable, $\mathcal{X}\otimes\V{\beta}$ cannot be of low rank, and $\mathbf{X}_k$'s cannot be collinear. Formally, we make the following assumptions:
\begin{enumerate}
\item[A1.]  $\mathrm{rank}(\mathcal{X}\otimes\V{\beta})>r$ for all $\V{\beta}\neq \mathbf{0}$;
\item[A2.]  $\mathrm{vec}(\mathbf{X}_1),\dots,\mathrm{vec}(\mathbf{X}_m)$ are linearly independent.
\end{enumerate}
Assumption A1 implies that $\mathcal{X}\otimes\V{\beta}$ is linearly independent of $\M{\Theta}$, and assumption A2 ensures that $\V{\beta}$ is identifiable.

\subsection{Estimation}\label{estimation}
In principle, estimates of $\mathbf{\Theta}$ and $\V{\beta}$ can be obtained by maximizing the constrained log-likelihood as follows
\begin{align}
(\overline{\mathbf{\Theta}},\overline{\V{\beta}})=\operatorname*{\arg\max}_{(\mathbf{\Theta},\V{\beta}):\mathrm{rank}(\mathbf{\Theta}) \leq r}\ell_{\mathbf{A},\mathcal{X}}(\mathbf{\Theta},\V{\beta}), \label{opt1}
\end{align}
where $\ell_{\mathbf{A},\mathcal{X}}$ is the log-likelihood based on the distribution in \eqref{lrem:canonical_form}. 
Note the distinction between directed and undirected networks is not crucial here because the estimators will automatically be symmetric when $\mathbf{X}_1,\dots,\mathbf{X}_m$ and $\mathbf{A}$ are symmetric.

Although in practice certain algorithms such as the alternating direction method may be applied to solve \eqref{opt1}, no computationally feasible algorithm is guaranteed to find the global maximum due to the non-convexity of the rank constraint $\mathrm{rank}(\M{\Theta}) \leq r$. To circumvent this, the rank constraint is often replaced with a convex relaxation \citep[e.g.][]{Candes2009}.   Let $\operatorname{conv}(\mathcal{S})$ denote the convex hull of set $\mathcal{S}$, $\sigma_i(\mathbf{\Theta})$ the $i$-th largest singular value of $\mathbf{\Theta}$, and $\|\mathbf{\Theta}\|_*$ the nuclear norm of $\mathbf{\Theta}$.  Then a common relaxation is 
\begin{align*}
& \operatorname{conv} \{\mathbf{\Theta}:\mathrm{rank}(\mathbf{\Theta})\leq r,\|\mathbf{\Theta}\|_2\leq 1\} \\
\ &=\operatorname{conv}\{\mathbf{\Theta}:\mathbf{\Theta}\mbox{ has at most $r$ non-zero singular values and }\sigma_i(\mathbf{\Theta})\leq 1~\forall i\} \\
\ &=\{\mathbf{\Theta}:\sum^n_{i=1}\sigma_k(\mathbf{\Theta})\leq r\}  =\{\mathbf{\Theta}:\|\mathbf{\Theta}\|_*\leq r\} . 
\end{align*}
 Using this relaxation, one can estimate $\mathbf{\Theta}$ and $\V{\beta}$ by solving the problem 
\begin{align}
(\widetilde{\mathbf{\Theta}},\widetilde{\V{\beta}})=\operatorname*{\arg\max}_{(\mathbf{\Theta},\V{\beta}):\|\mathbf{\Theta}\|_*\leq R}\ell_{\mathbf{A},\mathcal{X}}(\mathbf{\Theta},\V{\beta}), \label{opt2}
\end{align}
where $R$ is a tuning parameter.  The exponential family assumption and the use of the nuclear norm ensure the strict convexity of \eqref{opt2} as a function of $\theta_{ij}$'s  and therefore the uniqueness of the maximum. Finally, the mean matrix $\M{P}$ can be estimated by $\widetilde{\M{P}}=L^{-1}(\widetilde{\mathbf{\Theta}}+\mathcal{X}\otimes\widetilde{\V{\beta}})$.  

The optimization problem \eqref{opt2} can be solved by the standard projected gradient algorithm \citep{Boyd2009}. Specifically, the main (block-coordinate) updating formulas are
\begin{enumerate}
\item ${\beta}^{(t+1)}_k\gets \beta^{(t)}_k+\gamma_t\nabla_{\beta_k}\ell_{\mathbf{A},\mathcal{X}}(\mathbf{\Theta}^{(t)},\V{\beta})|_{\V{\beta}=\V{\beta}^{(t)}}$ for $k=1,\dots,m$
\item $\mathbf{\Theta}^{(t+1)}\gets \mathcal{P}\Big(\mathbf{\Theta}^{(t)}+\gamma_t\nabla_{\mathbf{\Theta}}\ell_{\mathbf{A},\mathcal{X}}(\mathbf{\Theta},\V{\beta}^{(t+1)})|_{\mathbf{\Theta}=\mathbf{\Theta}^{(t)}}\Big)$
\end{enumerate}
where $\gamma_t$ is a step size and $\mathcal{P}$ is a projection operator onto the set $\{\mathbf{\Theta}: \|\mathbf{\Theta}\|_*\leq R\}$. The first updating formula is the same as the standard gradient ascent algorithm since there is no constraint on $\V{\beta}$. The second formula consists of a gradient ascent step and a projection operation to ensure that the algorithm produces a solution in the feasible set. Thus, for solving \eqref{opt2}, we have
\begin{align}
\V{\beta}^{(t+1)}_k&\gets \V{\beta}^{(t)}_k+\gamma_t\Big(\mathrm{tr}\big(\mathbf{X}_k^\top(\mathbf{A}-L^{-1}(\mathbf{\Theta}^{(t)}+\mathcal{X}\otimes\V{\beta}^{(t)}))\big)\Big)\mbox{ for }k=1,\dots,m \notag \\
\mathbf{\Theta}^{(t+1)}&\gets \mathcal{P}_{c_t}\Big(\mathbf{\Theta}^{(t)}+\gamma_t\big(\mathbf{A}-L^{-1}(\mathbf{\Theta}^{(t)}+\mathcal{X}\otimes\V{\beta}^{(t+1)})\big)\Big), \label{lrem:opt2_update2}
\end{align}
where $\mathcal{P}_{c_t}(\mathbf{\Theta})=\sum^n_{i=1}(\sigma_i-{c_t})_+\mathbf{u}_i\mathbf{v}_i^\top$  $\mathcal{P}_{c_t}$ is a soft-thresholding operator, $\mathbf{\Theta}=\sum^n_{i=1}\sigma_i\mathbf{u}_i\mathbf{v}_i^\top$ is the singular value decomposition (SVD) of $\mathbf{\Theta}$, and $c_t=\operatorname{\arg\min}_c\{\sum^n_{i=1}(\sigma_i-c)_+\leq R\}$. 
Since the log-likelihood is continuously differentiable, convergence of the algorithm is guaranteed by choosing $\gamma_t<K^{-1}$ when the gradient of the log-likelihood is $K$-Lipschitz continuous on the feasible set. For example, in the case of the logit link, $K=1$, and for the logarithm link (when the edge weight follows a Poisson distribution), $K=\exp(\|\mathbf{\Theta}\|_{\max}+\max_{(i,j)}\boldsymbol{x}_{ij}^\top\V{\beta})$, where $\|\mathbf{\Theta}\|_{\max}$ denotes the maximum absolute entry of $\mathbf{\Theta}$.  See \cite{Boyd2009} for theoretical details and a variety of accelerated projected gradient algorithms. 

The updating formulas require solving a full SVD in each iteration, which can be computationally expensive, especially when $n$ is large.  In practice, if the matrix $\mathbf{\Theta}^{(t)}+\gamma_t\big(\mathbf{A}-L^{-1}(\mathbf{\Theta}^{(t)}+\mathcal{X}\otimes\V{\beta}^{(t+1)})\big)$ is approximately low rank, solving the SVD truncated at rank $s$ for some $s>r$ usually gives the same optimum.  Thus, we consider an alternative criterion to \eqref{opt2} to estimate $\mathbf{\Theta}$ and $\V{\beta}$, i.e.
\begin{align}
(\widehat{\mathbf{\Theta}},\widehat{\V{\beta}})=\operatorname*{\arg\max}_{(\mathbf{\Theta},\V{\beta}):\|\mathbf{\Theta}\|_*\leq R, \mathrm{rank}(\mathbf{\Theta})\leq s}\ell_{\mathbf{A},\mathcal{X}}(\mathbf{\Theta},\V{\beta}), \label{opt3}
\end{align}
and solve the optimization problem by replacing the nuclear-norm projection operator in \eqref{lrem:opt2_update2} with $\mathcal{P}_{(R,s)}=\sum^s_{i=1}(\sigma_i-c_t)_+\mathbf{u}_i\mathbf{v}_i^\top$, with $c_t$ as defined above.  Finally, the mean matrix $\M{P}$ is estimated by 
$$\widehat{\M{P}}=L^{-1}(\widehat{\mathbf{\Theta}}+\mathcal{X}\otimes\widehat{\V{\beta}}).$$
Although the optimization problem in \eqref{opt3} is non-convex, as illustrated in Figure \ref{fig_opt3}, the algorithm is computationally efficient, and we will also show that the estimator enjoys theoretical guarantees  similar to those of \eqref{opt2}.

\begin{figure}
\begin{center}
       \makebox[\textwidth]{\begin{subfigure}[b]{0.4\textwidth}
\begin{tikzpicture}[z=-0.5cm,->,thick, scale=0.5]
    \draw[-] (5, 0, 0)node[right]{$\sigma_1$};
    \draw[-] (0, 5, 0)node[right]{$\sigma_2$};
    \draw[-] (0, 0, 5)node[left]{$\sigma_3$};

    \draw[color=black] (0,0,0) -- (5,0,0);
    \draw[color=black] (0,0,0) -- (0,5,0);
    \draw[color=black] (0,0,0) -- (0,0,5);
    \draw[color=black!50, line width=1.5mm,-] (0,0,0) -- (3.5,0,0);
    \draw[color=black!50, line width=1.5mm,-] (0,0,0) -- (0,3.5,0);
    \draw[color=black!50, line width=1.5mm,-] (0,0,0) -- (0,0,3.5);

\end{tikzpicture}
                \caption{$\mathrm{rank}(\mathbf{\Theta})\leq 1$}
                \label{fig_opt1}
        \end{subfigure}
                       \begin{subfigure}[b]{0.4\textwidth}
\begin{tikzpicture}[z=-0.5cm,->,thick, scale=0.5]
    \draw[-] (5, 0, 0)node[right]{$\sigma_1$};
    \draw[-] (0, 5, 0)node[right]{$\sigma_2$};
    \draw[-] (0, 0, 5)node[left]{$\sigma_3$};

    \draw[color=black] (0,0,0) -- (5,0,0);
    \draw[color=black] (0,0,0) -- (0,5,0);
    \draw[color=black] (0,0,0) -- (0,0,5);
    \filldraw[ball color=black!30] (3.5,0,0) -- (0,3.5,0) -- (0,0,3.5) --cycle;
\end{tikzpicture}

                \caption{$\|\mathbf{\Theta}\|_*\leq R$}
                \label{fig_opt2}
        \end{subfigure}
                       \begin{subfigure}[b]{0.4\textwidth}
\begin{tikzpicture}[z=-0.5cm,->,thick, scale=0.5]
    \draw[-] (5, 0, 0)node[right]{$\sigma_1$};
    \draw[-] (0, 5, 0)node[right]{$\sigma_2$};
    \draw[-] (0, 0, 5)node[left]{$\sigma_3$};

    \draw[color=black] (0,0,0) -- (5,0,0);
    \draw[color=black] (0,0,0) -- (0,5,0);
    \draw[color=black] (0,0,0) -- (0,0,5);
    \filldraw[ball color=black!30] (0,0,0) -- (0,3.5,0) -- (0,0,3.5) --cycle;
    \filldraw[ball color=black!30] (0,0,0) -- (3.5,0,0) -- (0,0,3.5) --cycle;
    \filldraw[ball color=black!30] (0,0,0) -- (0,3.5,0) -- (3.5,0,0) --cycle;
\end{tikzpicture}
                \caption{$\mathrm{rank}(\mathbf{\Theta})\leq 2$ and $\|\mathbf{\Theta}\|_*\leq R$}
                \label{fig_opt3}
        \end{subfigure}}

        \end{center}
\caption{Constraints in optimization problems (\ref{opt1}), (\ref{opt2}), and (\ref{opt3}) in the space of singular values of $\mathbf{\Theta}$}
\end{figure}
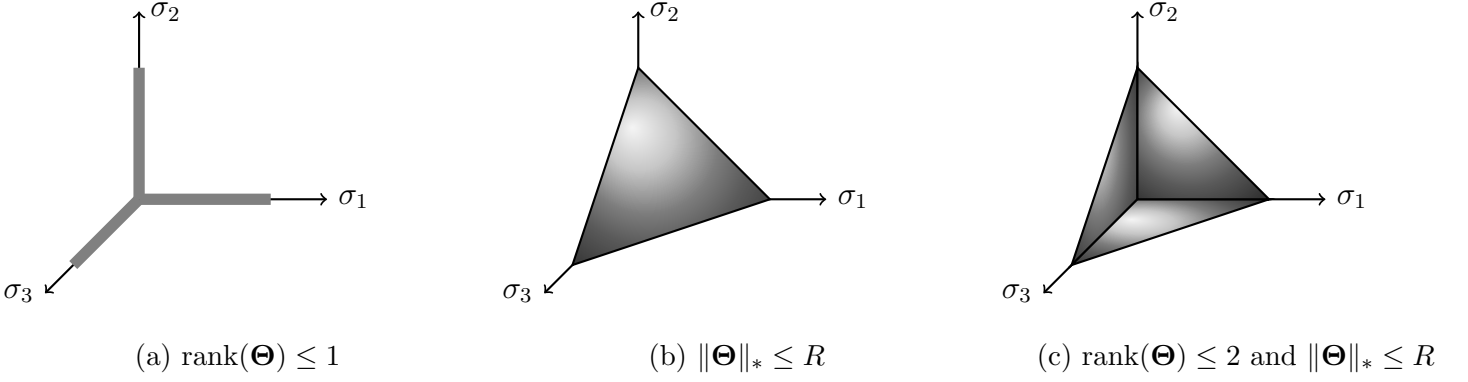

\section{Theoretical properties}\label{sec_consistency}
In this section, we show asymptotic properties of our estimates for the low rank GLM,  in Frobenius matrix norm.
We make the following additional assumptions on the parameter space and covariates:

\begin{itemize}
\item[A3.] 
$\|\mathbf{\Theta}\|_{\max}\leq K_\theta ~\mbox{and}~ \mathrm{rank}(\mathbf{\Theta})\leq r$ 
\item[A4.]   $\|\V{\beta}\|_2 \leq K_\beta$ 
\item[A5.]  $\|\boldsymbol{x}_{ij}\|_2\leq K_x$ for all $i,j$
\end{itemize}

\begin{theorem}\label{thm1}
Under assumptions A3-A5, we have
\[n^{-1}\|\widetilde{\mathbf{P}}-\mathbf{P}\|_F\stackrel{p}{\longrightarrow}0,\]
where $\widetilde{\mathbf{P}}=L^{-1}(\widetilde{\mathbf{\Theta}}+\mathcal{X}\otimes\widetilde{\V{\beta}})$, and $\widetilde{\mathbf{\Theta}}$ and $\widetilde{\V{\beta}}$ are obtained from (\ref{opt2}).
\end{theorem}
Similarly, consistency of $\widehat{\mathbf{P}}$ can also be established.
\begin{corollary}\label{cor0}
Under assumptions A3-A5, we have, 
\[n^{-1}\|\widehat{\mathbf{P}}-\mathbf{P}\|_F\stackrel{p}{\longrightarrow}0,\]
where $\widehat{\mathbf{P}}=L^{-1}(\widehat{\mathbf{\Theta}}+\mathcal{X}\otimes\widehat{\V{\beta}})$, and $\widehat{\mathbf{\Theta}}$ and $\widehat{\V{\beta}}$ are obtained from (\ref{opt3}).
\end{corollary}

The tail probabilities of both $n^{-1}\|\widetilde{\mathbf{P}}-\mathbf{P}\|_F$ and $n^{-1}\|\widehat{\mathbf{P}}-\mathbf{P}\|_F$ have a polynomially-decaying rate. We can obtain a better probability bound for some widely-used models such as logit models as stated in the following corollary.  
\begin{corollary}\label{cor1}
Under the assumptions of Theorem \ref{thm1},  if $A_{ij}$'s are uniformly bounded, then both $n^{-1}\|\widetilde{\mathbf{P}}-\mathbf{P}\|_F$ and $n^{-1}\|\widehat{\mathbf{P}}-\mathbf{P}\|_F$ have an exponentially-decaying tail probability. 
\end{corollary}

Beyond $\widehat{\mathbf{P}}$, asymptotic properties of $\widehat{\M{\Theta}}$ and $\widehat{\V{\beta}}$ are also often of interest.  If they are identifiable, the following corollary gives consistency for the parameters.  
\begin{corollary}\label{col2}
If assumptions A1-A5 hold, $\inf_{ij}\mathrm{Var}(A_{ij})>0$, and there exists $0<\delta<1$ such that
\begin{align}
\sup_{\V{\beta}}\frac{\sum^{r+s}_{i=1}\sigma_i^2(\mathcal{X}\otimes\V{\beta})}{\sum^n_{i=1}\sigma_i^2(\mathcal{X}\otimes\V{\beta})}\leq \delta < 1,  \label{eff_rank}
\end{align}
then 
\begin{align*}
n^{-1}\|\widehat{\mathbf{\Theta}}-\mathbf{\Theta}\|_F  & \stackrel{P}{\longrightarrow} 0 \\
\|\widehat{\V{\beta}}-\V{\beta}\|_F &  \stackrel{P}{\longrightarrow} 0.
\end{align*}
\end{corollary}
Note that we can drop the supermum in the condition \eqref{eff_rank} as $\frac{\sum^{r+s}_{i=1}\sigma_i^2(\mathcal{X})}{\sum^n_{i=1}\sigma_i^2(\mathcal{X})}\leq\delta<1$ if $\V{\beta}$ is univariate and correspondingly $\mathcal{X}$ is a matrix.  

The convex relaxation in \eqref{opt2} changes the feasible set, and in the new parameter space, $(\V{\beta}$ and $\mathbf{\Theta})$ may no longer be identifiable. Therefore consistency of $\widetilde{\V{\beta}}$ and $\widetilde{\mathbf{\Theta}}$ is not guaranteed.


A case of practical interest is when $\mathbf{\Theta}$ is only approximately rather than exactly low rank (i.e., has a few large leading eigenvalues and the other eigenvalues are relatively small but not necessarily 0).   We can then show the bias of $\widetilde{\mathbf{P}}$ and $\widehat{\mathbf{P}}$ caused by model misspecification can be bounded as follows.  
\begin{theorem}\label{thm2}
Under the assumptions of Theorem \ref{thm1}, except that $\mathrm{rank}(\mathbf{\Theta})>r$, we have
\[\frac{\|\widetilde{\mathbf{P}}-\mathbf{P}\|_F^2}{\sum^n_{k=r+1}\sigma_k(\mathbf{\Theta})}=O_P(1),\]
and
\[\frac{\|\widehat{\mathbf{P}}-\mathbf{P}\|_F^2}{\sum^n_{k=r+1}\sigma_k(\mathbf{\Theta})}=O_P(1).\]
\end{theorem}

This result suggests that our proposed estimates enjoy robustness under model misspecification if the eigenvalues following the first $r$ are small.  This holds even if  $r$ grows with $n$  as long as $r=o(n)$. As an application of Theorem \ref{thm2}, we present the error bound for the low rank effects model for binary networks as an example.
\begin{example}[Bias of the low rank effects logistic model]
For logistic models, $b(\eta)=\log(1+e^\eta)$. Thus, by \eqref{lrem:misspecification} in the proof of Theorem \ref{thm2}, we have
\begin{align*}
\sum_{ij}  & \big(b(\widehat{\theta}_{ij}+\mathbf{x}_{ij}^\top\V{\beta}) - b(\theta_{ij}+\mathbf{x}_{ij}^\top\V{\beta})\big) \\
& = \sum_{ij}\log\bigg(\frac{1+e^{\widehat{\theta}_{ij}+\mathbf{x}_{ij}^\top\beta}}{1+e^{\theta_{ij}+\mathbf{x}_{ij}^\top\beta}}\bigg) \leq n\sum^n_{k=r+1}\sigma_k(\mathbf{\Theta})
\end{align*}
and therefore 
\[\mathbb{P}\bigg(n^{-1}\|\widehat{\mathbf{P}}-\mathbf{P}\|_F\leq\bigg(2n^{-1}\sum^n_{k=r+1}\sigma_k(\mathbf{\Theta})\bigg)^{\frac{1}{2}}\bigg)\rightarrow 1.\]
\end{example}

%

\section{Results on synthetic networks}\label{sec_simulation}
In this section, we present numerical results on simulated data to demonstrate the finite sample performance of the proposed low rank effects model and compare to benchmark methods.    For the sake of computational efficiency, we focus on the estimate given by \eqref{opt3}.

We consider a generative model similar to \eqref{proj}, with the mean function given by
\begin{align}
L(\mathbf{P})=\mathbf{Z}\mathbf{Z}^\top + \alpha\mathbf{11}^\top + \mathcal{X}\otimes\boldsymbol{\beta}, \label{gen_model}
\end{align}
where $\mathbf{Z}\sim [N(0,1)]_{n\times (r-1)}$ with independent entries. For the feature tensor $\mathcal{X}=[\mathbf{X}_1,\mathbf{X}_2]_{n\times n\times 2}$, we first generate $\widetilde{\mathbf{X}}\sim [N(0,1)]_{n\times n}$ with independent entries and then compute $\mathbf{X}=\mathbf{UV}^\top$, where $\mathbf{U}$ and $\mathbf{V}$ are obtained from $\widetilde{\mathbf{X}}\stackrel{SVD}{=}\mathbf{UDV}^\top$. Therefore, all singular values of both $\mathbf{X}_1$ and $\mathbf{X}_2$ are equal to 1. Specifically,  so they are full rank. We set $n=200$ and $r=2$,  and $\boldsymbol{\beta}=(c,-c)$.  Given the mean function $L$ and $\mathcal{X}$ and $\mathbf{Z}$, we generate conditionally independent edges.  We vary the parameters $\alpha$ and $c$ to investigate the density of the network and the relative importance of low rank effects and covariates. 

As benchmarks, we fit the classical GLMs and latent models, with details given below.  The estimation for latent models is based on 500 burn-in and 10,000 MCMC iterations in each setting.  Following the evaluation method for link prediction  in \cite{Zhao2013}, all tuning parameters for the low rank effects model and latent models are selected with subsampling validation.  Specifically, we create training data networks by setting randomly selected 20\% of all edges to 0, and calculate the predictive area under the ROC curve (AUC), which is defined as
\[ \mathrm{AUC}(\mathbf{A},\widehat{\mathbf{P}}) = \frac{\sum_{(i,j),(i',j')\in\mathcal{I}} 1(A_{ij}=0, A_{i'j'} > 0, \widehat{p}_{ij}<\widehat{p}_{i'j'})}{\sum_{(i,j),(i',j')\in\mathcal{I}} 1(A_{ij}=0,A_{i'j'}>0)},\]
where $\mathcal{I}$ is the index set of the ``held-out'' edges. With the selected tuning parameter, we fit the model to the entire network to obtain $\widehat{\mathbf{P}}$.   We then generate test networks $\mathbf{A}_{test}$ and compute $\mathrm{AUC}(\mathbf{A}_{test},\widehat{\mathbf{P}})$. 
In simulation studies, we have also computed the relative mean squared error for $\M{P}$, defined as $\mathrm{RMSE}(\widehat{\mathbf{P}})=\|\widehat{\mathbf{P}}-\mathbf{P}\|_F / \|\mathbf{P}\|_F$.


\subsection{Binary networks}\label{lrem:sim:binary}
By setting $L(p)=\mathrm{logit}(p)$ in model \eqref{gen_model}, we generated directed binary networks, with edges conditional on parameters generated independent Bernoulli random variables.  For each training network, we generated 10 test networks using the same parameters and covariates to evaluate the predictive AUC. For each setting, we also computed the RMSE.  The logistic regression model and the latent factor model \citep{Hoff2008} were used as benchmarks. 

Average results over 100 replications are shown in Figures \ref{fig_logit_auc} and \ref{fig_logit_mse}.  Although the low rank effects model (LREM) has a somewhat larger parameter RMSE when the networks are sparse (small values of $\alpha$), it outperforms both logistic regression and the latent factor model in terms of predictive AUC.    When the value of $c$ is large, most of the signal comes from the covariates rather than the low rank effects, and thus LREM behaves similarly to  logistic regression.  However, when the value of $c$ is small, LREM outperforms logistic regression, especially on predictive AUC, by properly combining the information from both the network and the node covariates.     We also observed that our algorithm produced much more numerically stable results than the latent factor model, with vastly lower computational cost. For example, in this simulation, for each setting our algorithm can converge in a few minutes on one single laptop.

\begin{figure}
\begin{center}
\includegraphics[width=0.8\textwidth]{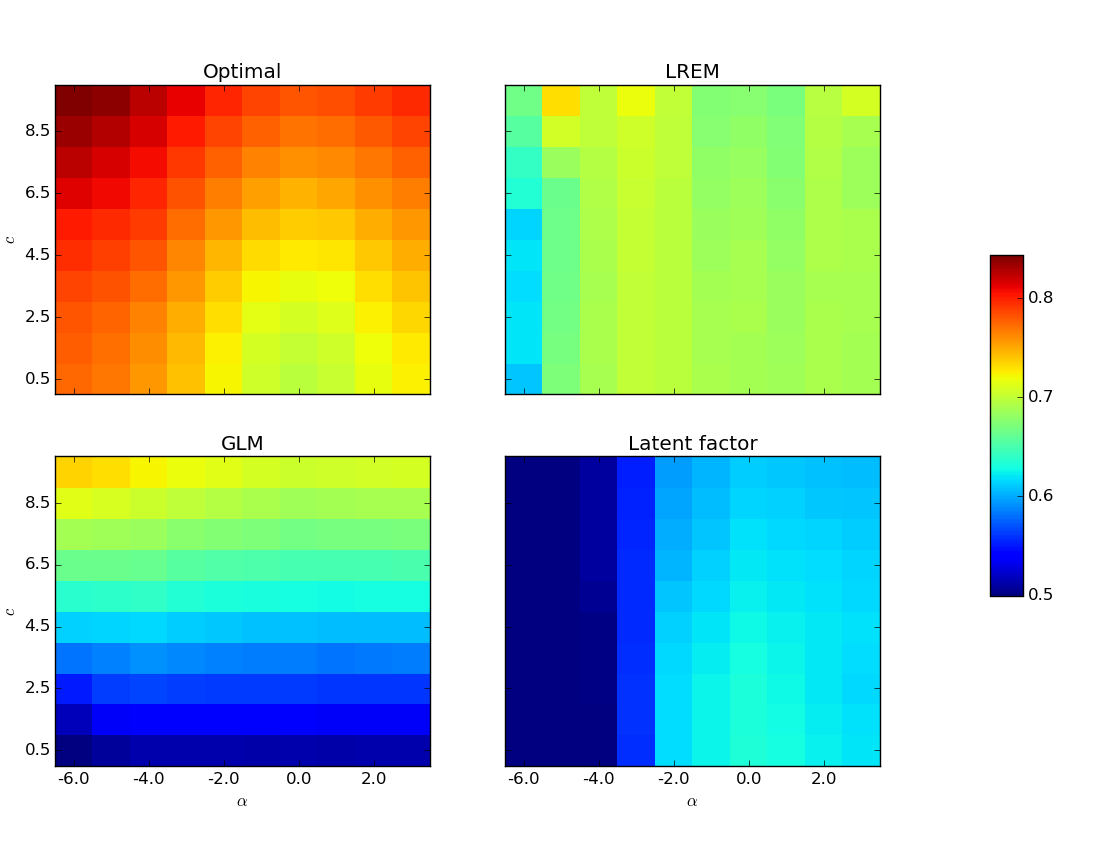}
\caption{Predictive AUC for binary networks with various $\alpha$ and $c$. ``Optimal'' is the AUC based on the true mean matrix $\mathbf{P}$.}\label{fig_logit_auc}
\end{center}

\begin{center}
\includegraphics[width=0.8\textwidth]{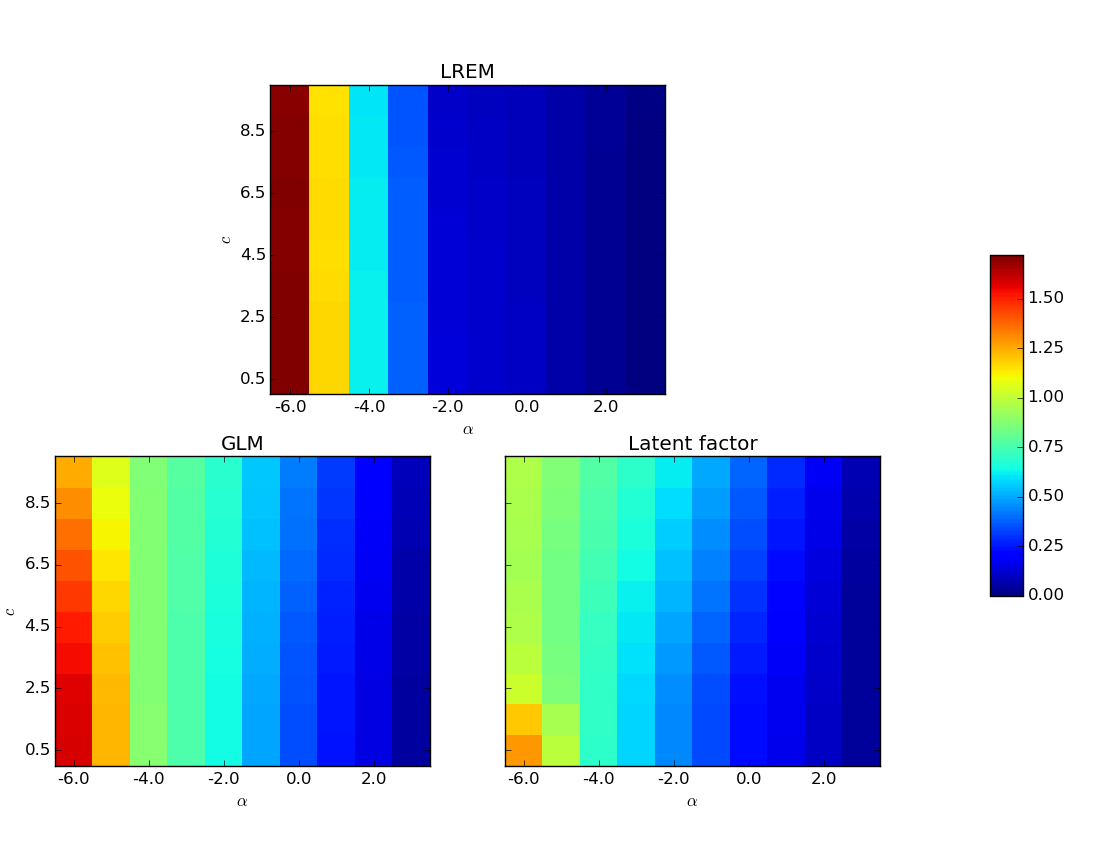}
\caption{RMSE for $\M{P}$ for binary networks with various $\alpha$ and $c$. }\label{fig_logit_mse}
\end{center}
\end{figure}

\subsection{Integer-weighted networks}
An important advantage of the proposed low rank effects model is that it extends trivially to weighted networks. Using the link function $L(p)=\log p$, we generated networks based on \eqref{gen_model} with edges conditionally independent Poisson random variables.  All other aspects of the simulation remain the same.  We consider the Poisson model and the fixed rank nomination model \citep{Hoff2012} for integer-weighted networks as benchmarks.  Note the fixed rank nomination model was originally developed for networks with partial rank ordering relationships, but since integer weights can be viewed as the strength of relationships in this model, it is a natural benchmark for comparison.  Since the AUC cannot be readily calculated on non-binary networks, we measure the performance based on ``classifying'' pairs of nodes that are connected ($A_{ij}>0$) versus not connected ($A_{ij}=0$). 

Average results over 100 replications are shown in Figures \ref{fig_poisson_auc} and \ref{fig_poisson_mse}.    In terms of predictive AUC, which is more relevant in practice, the low rank effects model substantially outperforms the Poisson model and the fixed rank nomination model, except for the largest values of $\alpha$ where the fixed rank nomination model performs slightly better.  

For the RMSE, the low rank effects model performs much better for all but the largest values of both $c$ and $\alpha$, which correspond to dense networks with high variation in node degrees.  In this setting for integer-weighted networks, one needs a larger sample size in order to obtain an accurate estimate of $\mathbf{P}$, which is consistent with the theoretical results in Theorem \ref{thm1} and Corollary \ref{cor1}.  

\begin{figure}
\begin{center}
\includegraphics[width=0.8\textwidth]{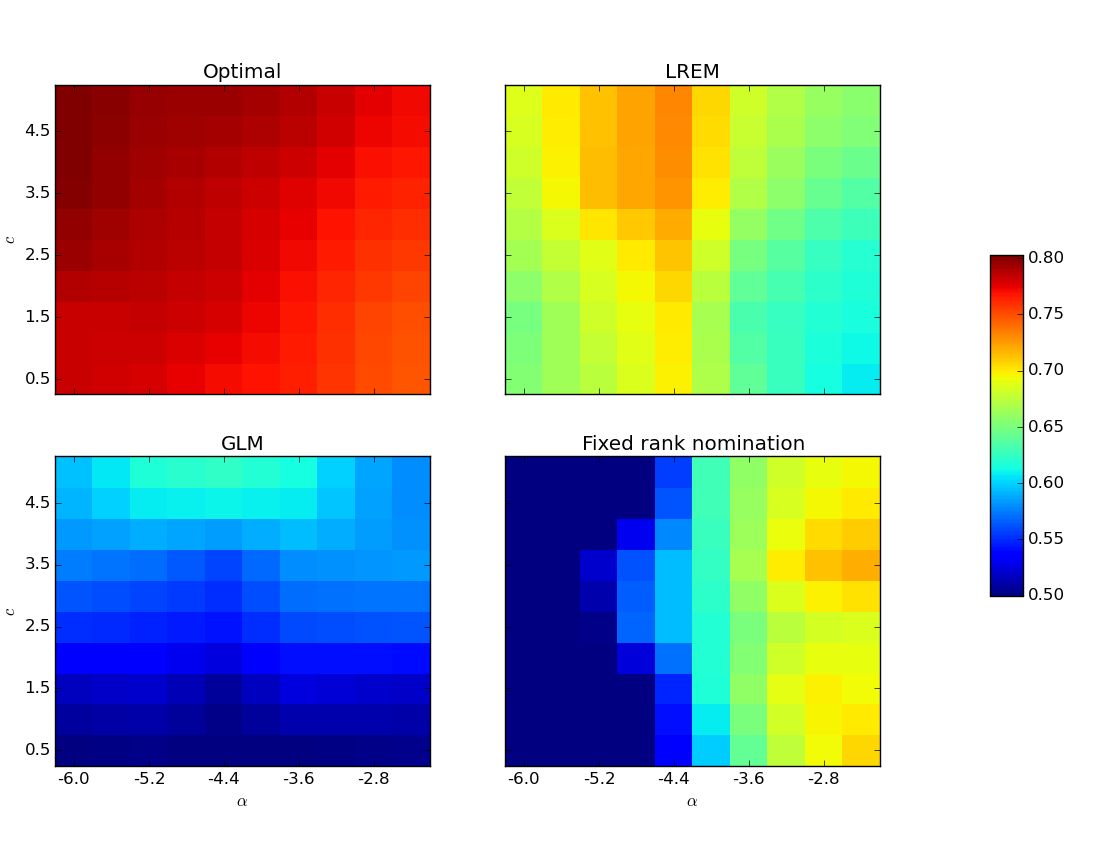}
\caption{Predictive AUC for integer-weighted weighted networks with various $\alpha$ and $c$. ``Optimal'' refers to the AUC based on the true mean matrix $\mathbf{P}$.}\label{fig_poisson_auc}
\end{center}

\begin{center}
\includegraphics[width=0.8\textwidth]{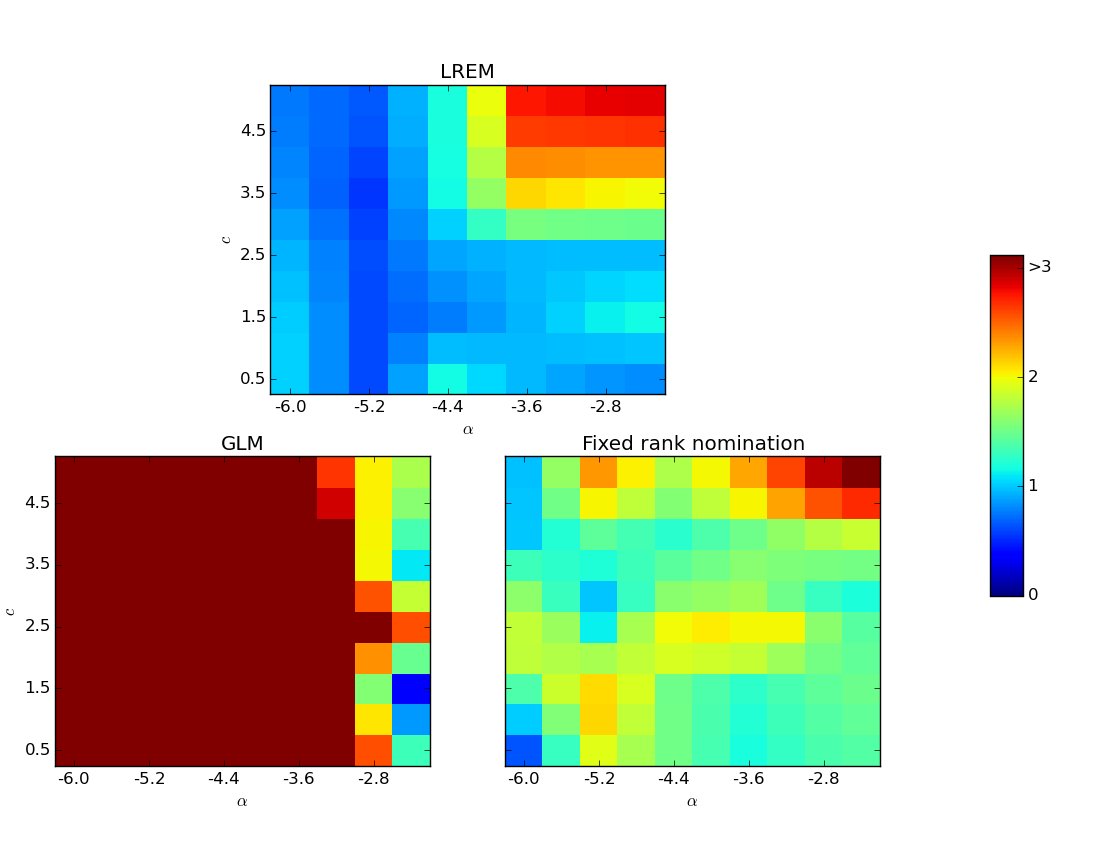}
\caption{RMSE for $\M{P}$ for integer-weighted weighted networks with various $\alpha$ and $c$. }\label{fig_poisson_mse}
\end{center}
\end{figure}

\section{Data examples}\label{sec_data_analysis}
Next, we apply the proposed low rank effects model to two real-world datasets.  To evaluate the performance, we randomly set 20\% of the entries in the adjacency matrices to 0 and compute the predictive AUC on this ``hold-out'' set.  This evaluation mechanism corresponds to the setting of partially observed networks discussed in \cite{Zhao2013}.  Reported results are averages over 20 repetitions.

\subsection{The Last.fm friendship data}
This dataset from the Last.fm music website friendships and 17,632 artists listened to or tagged by each user \citep{Cantador2011}. The friendship network contains 1,892 nodes (users) and 12,717 edges.  We constructed two edge attributes $\mathbf{X}_{\mathrm{lis(ten)}}$ and $\mathbf{X}_{\mathrm{tag}}$ as follows: let $\widetilde{X}_{\mathrm{lis},ij}$ and $\widetilde{X}_{\mathrm{tag},ij}$ be the number of artists who are listened to and tagged by, respectively, both users $i$ and $j$.    These counts were then normalized, setting $X_{\mathrm{lis},ij}=\widetilde{X}_{\mathrm{lis},ij} / \max_{ij}\{\widetilde{X}_{\mathrm{lis},ij}\}$ and $X_{\mathrm{tag},ij}=\widetilde{X}_{\mathrm{tag},ij} / \max_{ij}\{\widetilde{X}_{\mathrm{tag},ij}\}$.

The prediction results are shown in Figure \ref{fig_lastfm}.  The low rank effects model with covariates obtains the best AUC value of 0.876 at $r=42$ and $R=470$.  Although this value of AUC is likely to be overly optimistic, note that the predictive AUC of the low rank effects model is larger than 0.75 over the entire range of parameters $r$ and $R$, where as the logistic regression model only gives the AUC of 0.412.    This suggests that modeling low rank pairwise effects is important for this dataset.    The latent factor model (implemented via the package \texttt{amen} in R) failed to converge due to the size of the dataset. 

In Figure \ref{fig_lastfm_beta}, both $\widehat{\beta}_{\mathrm{lis}}$ and $\widehat{\beta}_{\mathrm{tag}}$ are positive and indicate that the Last.fm friendship network likely follows the principle of homophiliy. The rank constraint $r$ has very little effect on the estimates of the coefficients, while the estimates shrink toward 0 as the nuclear-norm constraint $R$ decreases due to the bias caused by a small $R$ and the fact that $\|\widehat{\boldsymbol{\Theta}}\|_{\max}\leq\|\widehat{\boldsymbol{\Theta}}\|_*\leq R$.

\begin{figure}
\begin{center}
\includegraphics[width=0.8\textwidth]{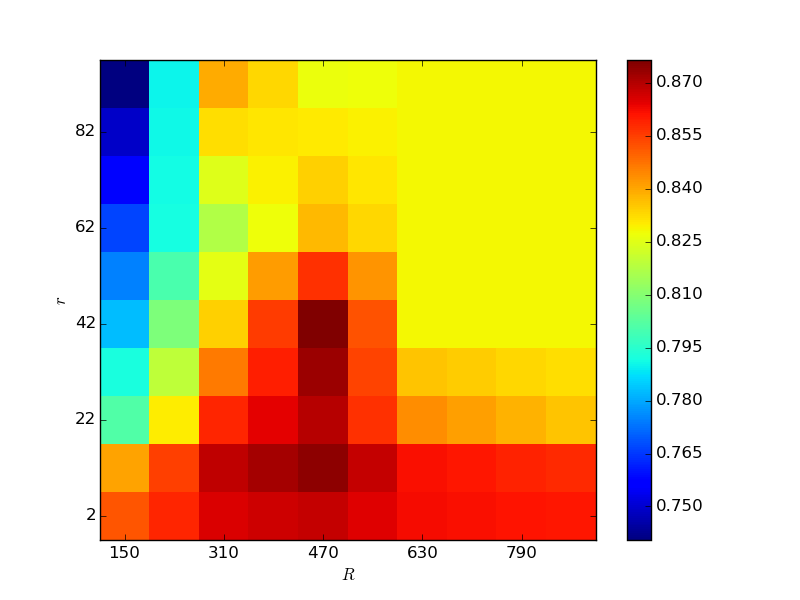}
\caption{Predictive AUC for the Last.fm music dataset with various tuning parameters $r$ and $R$.  }\label{fig_lastfm}
\end{center}

\begin{center}
\begin{subfigure}[b]{0.48\textwidth}
\includegraphics[width=1\textwidth]{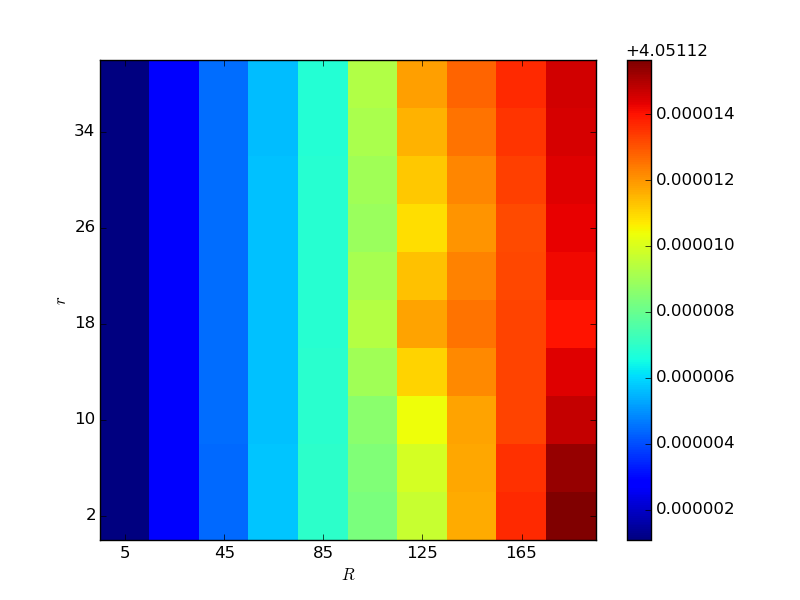}
\caption{$\widehat{\beta}_{\mathrm{lis}}$}
\end{subfigure}
\begin{subfigure}[b]{0.48\textwidth}
\includegraphics[width=1\textwidth]{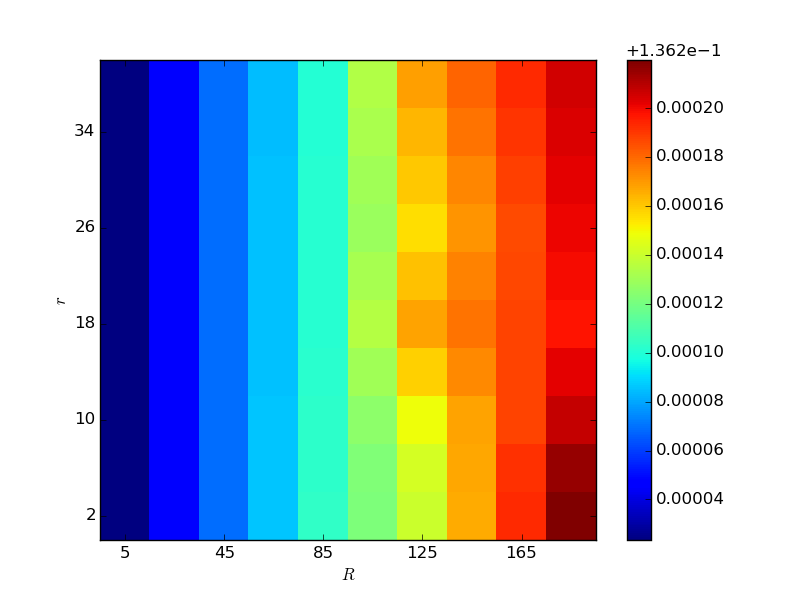}
\caption{$\widehat{\beta}_{\mathrm{tag}}$}
\end{subfigure}
\caption{Estimated coefficients for the Last.fm music dataset with various tuning parameters $r$ and $R$.}\label{fig_lastfm_beta}
\end{center}
\end{figure}

\subsection{The Elegans neural network data}
This dataset contains 
the neural network of the nematode worm C. elegans, which is a directed integer-weighted network with 297 nodes. In this network, an edge represents a synapse or a gap junction between two neurons \citep{Watts&Strogatz1998}, and the weight between a pair of nodes is the number of edges between two neurons. The mean weight is 29.69 and 2.66\% of pairs have non-zero weights. The original dataset does not contain any covariates. Therefore, we did not consider the classical GLM here, and the fixed rank nomination model was used as the benchmark model. Similar to the simulation studies for integer-weighted networks, we calculated the AUC based on ``classifying'' connected versus non-connected pairs.  

Figure \ref{fig_neural} shows the results from the low rank effects model.  The AUC obtains the maximum 0.824 at $r=26$ and $R=85$, which is roughly the same as the best performance of the fixed rank nomination model (AUC=0.821, fitted by 1,000 burn-in and 20,000 MCMC iterations, which is vastly more expensive computationally). The relatively high value of AUC indicates that there might be a low-rank effect associated with the observed network.

\begin{figure}
\begin{center}
\includegraphics[width=0.8\textwidth]{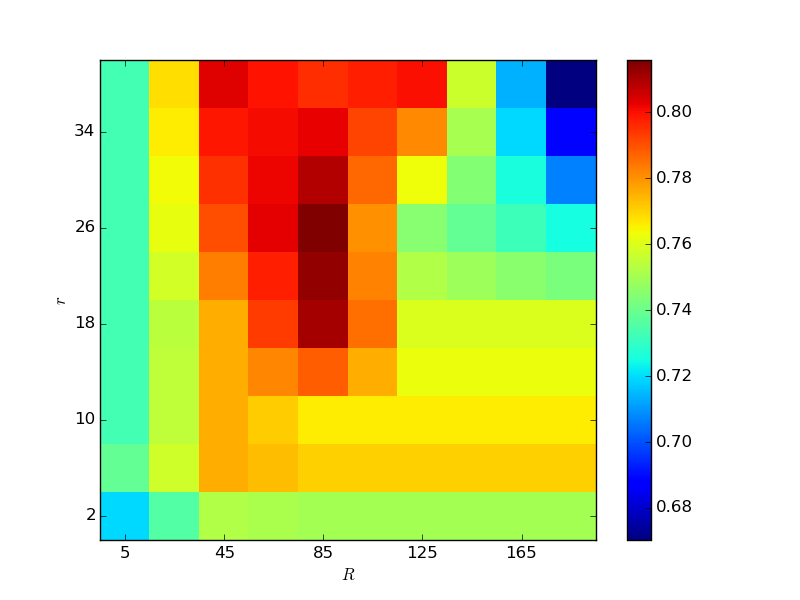}
\caption{Predictive AUC for the neural dataset with various values of  $r$ and $R$.}\label{fig_neural}
\end{center}
\end{figure}
\section{Discussion}\label{sec_discussion}
We proposed a generalized linear model with low-rank effects for network data with covariates, and an efficient projected gradient descent algorithm to fit this model.    The model is more general than the various latent space models\citep{Hoff2002, Hoff2007, Hoff2008, Ma2017}  because we do not require the effect matrix to be positive definite or symmetric, allowing for more general graph structures like bipartite graphs, and incorporating the directed case automatically.   The simultaneous work of \cite{Ma2017}   is the only scalable algorithm we are aware of for fitting relatively general latent space models, but it is still less general than ours;  and all previous work relied on MCMC and did not scale well at all.


  Figure \ref{fig_time} shows a simple comparison between the computational cost of our method and that of the latent factor model, for the simulation settings in this section.  For both methods, we show the relative cost for fitting binary networks described in Section \ref{lrem:sim:binary}. Compared to the case of $n=200$, it takes about 40 times of computational time for fitting the case of $n=2000$ for our method and about 120 times for the latent factor model.  The latent factor model becomes not feasible for networks with $10^5$ or more nodes.   

\begin{figure}
\begin{center}
\includegraphics[width=0.7\textwidth]{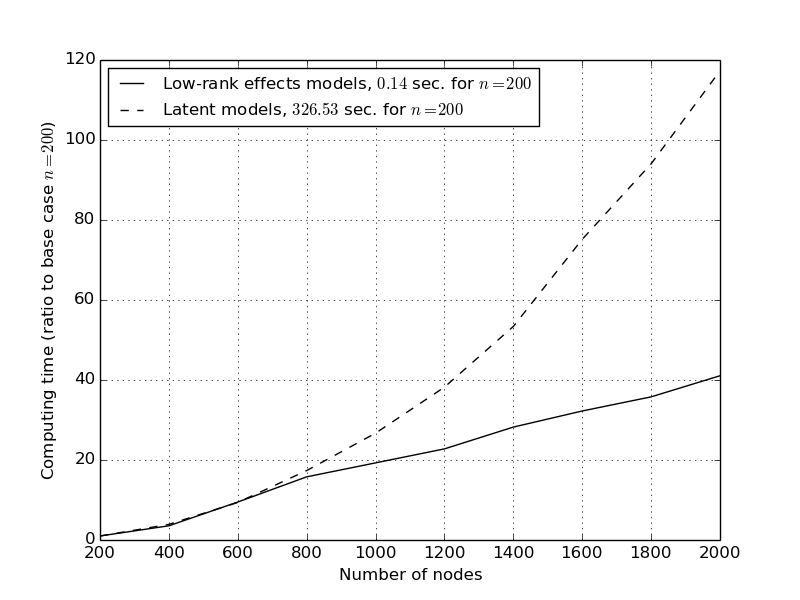}
\caption{Comparison of the computing time of the low rank effects model (using Python) and that of the latent factor model (using the R package \texttt{amen}), relative to their computing time when $n=200$. }\label{fig_time}
\end{center}
\end{figure}

There are several directions of future work to explore.  Any algorithm based on the SVD is in general considered not scalable to very large networks.  Boosting the computational speed of SVD-based algorithms usually relies on the sparsity of decomposed matrices, which does not apply to the low rank effects model even if the data network is sparse.  An alternative approach is the alternating direction method, which may find the global optimum when the estimator is obtained by minimizing the squared error loss under constraints.  However, generalizing the algorithm to the GLM setting is not trivial.   A stochastic gradient descent approach can also be applied to improve scalability.   

An obvious extenstion in the setting of high-dimensional covariates is to incorporate variable selection via penalties on $\V{\beta}$.  It should also be relatively straightforward to adapt this framework to modeling dynamic networks, where different networks are observed at different time points, with an underlying smoothly changing low rank probability matrix structure.      


\bibliographystyle{apalike}
\bibliography{ref}

\appendix
\section{Proof of theorems}
To establish consistency in Frobenius norm, we first state an inequality connecting the Frobenius norm to the Kullback-Leibler (KL) divergence, defined as
\begin{align*}
D_{KL}(f_{\mathbf{Q}_1}\|f_{\mathbf{Q}_2})&=n^{-2}\sum_{ij}\int^\infty_{-\infty} f_{q_{1,ij}}(a)\log\frac{f_{q_{1,ij}}(a)}{f_{q_{2,ij}}(a)}da,
\end{align*} 
where $\mathbf{Q}_1$ and $\mathbf{Q}_2$ are $n \times n$ matrices and $f_{\mathbf{Q}_1}$ and $f_{\mathbf{Q}_2}$ are the probability distributions of random matrices 
with mean $\mathbf{Q}_1$ and $\mathbf{Q}_2$ 
as defined in (\ref{lrem:canonical_form}).

Note that as a consequence of A3-A5, the $\xi$-th moment of $|A_{ij}|$ is uniformly bounded by some constant for each $\xi$, denoted by $M_\xi$, which does not depend on $n$.  Then using the uniform integrability given by the bounded parameter space, we have the following lemma.
\begin{lemma}\label{fro_tv_ineq}
Under assumptions A3-A5, we have
\[n^{-1}\|\mathbf{Q}_1-\mathbf{Q}_2\|_F\leq \sqrt{2}M_{1+\delta}^{\frac{1}{1+\delta}}D_{KL}^{\frac{\delta}{2+2\delta}}(f_{\mathbf{Q}_1}\|f_{\mathbf{Q}_2})\]
for some $\delta>0$.
\end{lemma}
\begin{proof}[Proof of Lemma \ref{fro_tv_ineq}]
Let \[\|f_{q_{1,ij}}-f_{q_{2,ij}}\|_{TV}=\sup_{g_{ij}:\mathbb{R}\rightarrow[-1,1]}\int g_{ij}(a)(f_{p_{ij}}(a)-f_{q_{ij}}(a))d\mu(a),\] 
where $\mu$ is the Lebesgue or counting measure. Then,
\begin{align*}
&\quad \|\mathbf{Q}_1-\mathbf{Q}_2\|_F^2 \\
&\leq\sum_{ij}\bigg(\int^\infty_0 |a||f_{q_{1,ij}}(a)-f_{q_{2,ij}}(a)|d\mu(a)\bigg)^2\\
&\leq\sum_{ij}\bigg(u_{ij}\int^{u_{ij}}_0|f_{q_{1,ij}}(a)-f_{q_{2,ij}}(a)|d\mu(a) +u_{ij}^{-t}\int^\infty_{u_{ij}} |a|^{1+\delta}(f_{q_{1,ij}}(a)+f_{q_{2,ij}}(a))d\mu(a)\bigg)^2 \\
&\leq\sum_{ij}\bigg(2u_{ij}\|f_{q_{1,ij}}-f_{q_{2,ij}}\|_{TV} +2u_{ij}^{-t}M_{1+\delta}\bigg)^2
\end{align*}
As a function of $u_{ij}$, the minimum of $u_{ij}\|f_{q_{1,ij}}-f_{q_{2,ij}}\|_{TV} +u_{ij}^{-\delta}M_{1+\delta}$  is obtained by choosing $u_{ij}=\delta^\frac{1}{1+\delta}M_{1+\delta}^{\frac{1}{1+\delta}}\|f_{q_{1,ij}}-f_{q_{2,ij}}\|_{TV}^{-\frac{\delta}{1+\delta}}$ and so
\begin{align*}
n^{-2}\|\mathbf{Q}_1-\mathbf{Q}_2\|_F^2&\leq n^{-2}\sum_{ij}(\delta^{\frac{1}{1+\delta}}+\delta^{-\frac{\delta}{1+\delta}})^2M_{1+\delta}^{\frac{2}{1+\delta}}\|f_{q_{1,ij}}-f_{q_{2,ij}}\|_{TV}^{\frac{2\delta}{1+\delta}} \\
&\leq 4n^{-2}M_{1+\delta}^{\frac{2}{1+\delta}}\sum_{ij}\|f_{q_{1,ij}}-f_{q_{2,ij}}\|_{TV}^{\frac{2\delta}{1+\delta}} \\
&\leq 2M_{1+\delta}^{\frac{2}{1+\delta}}D_{KL}^{\frac{\delta}{1+\delta}}(f_{\mathbf{Q}_1}\|f_{\mathbf{Q}_2})
\end{align*}
for any $\delta>0$. The last inequality is given by Pinsker's inequality.
\end{proof}

\begin{proof}[Proof of Theorem \ref{thm1}]
We define a feasible set of $(\mathbf{\Theta},\boldsymbol{\beta})$ as
\begin{align*}
\mathcal{T}&=\{(\mathbf{\Theta},\boldsymbol{\beta}):\|\mathbf{\Theta}\|_*\leq \sqrt{r}nK_\theta,\|\boldsymbol{\beta}\|_2\leq K_\beta\},
\end{align*}
and a corresponding estimator as
\begin{align}
(\widetilde{\mathbf{\Theta}},\widetilde{\boldsymbol{\beta}})=\operatorname*{\arg\max}_{(\mathbf{\Theta},\boldsymbol{\beta})\in\mathcal{T}}\ell_{\mathbf{A},\mathcal{X}}(\mathbf{\Theta},\boldsymbol{\beta}). \label{opt4}
\end{align}
Note that when $R=\sqrt{r}nK_\theta$ and $K_\beta$ is large enough, the solution for (\ref{opt4}) is the same as that for (\ref{opt2}).
Let $h(\mathbf{B},\mathbf{c}):=\mathbb{E}[\ell_{\mathbf{A},\mathbf{X}}(\mathbf{B},\mathbf{c})]$. Note that the maximum likelihood criterion in (\ref{opt3}) ensures that $\ell_{\mathbf{A},\mathbf{X}}(\widehat{\mathbf{\Theta}},\widehat{\boldsymbol{\beta}})\geq \ell_{\mathbf{A},\mathbf{X}}(\mathbf{\Theta},\boldsymbol{\beta})$. Hence, we have
\begin{align}
 n^2D_{KL}(f_{\mathbf{P}}\|f_{\widehat{\mathbf{P}}})&=h(\mathbf{\Theta},\boldsymbol{\beta})-h(\widehat{\mathbf{\Theta}},\widehat{\boldsymbol{\beta}}) \notag\\
&\leq \ell_{\mathbf{A},\mathbf{X}}(\widehat{\mathbf{\Theta}},\widehat{\boldsymbol{\beta}})- \ell_{\mathbf{A},\mathbf{X}}(\mathbf{\Theta},\boldsymbol{\beta}) +h(\mathbf{\Theta},\boldsymbol{\beta})-h(\widehat{\mathbf{\Theta}},\widehat{\boldsymbol{\beta}}) \notag\\
&=\mathrm{tr}((\mathbf{A}-\mathbf{P})^\top(\widehat{\mathbf{\Theta}}-\mathbf{\Theta})) \notag \\
&\quad +\sum^m_{k=1}(\widehat{\beta}_k-\beta_k)\mathrm{tr}((\mathbf{A}-\mathbf{P})^\top\mathbf{X}_k).  \label{ineq1}
\end{align}

To see the vanishing of the first term as $n$ goes to infinity,   one can derive that
\begin{align*}
\mathrm{tr}((\mathbf{A}-\mathbf{P})^\top(\widehat{\mathbf{\Theta}}-\mathbf{\Theta}))&\leq 2\sup_{\mathbf{\Xi}\in\mathcal{T}}|\mathrm{tr}((\mathbf{A}-\mathbf{P})^\top\mathbf{\Xi})|  \\ 
& \leq 2\sigma_1(\mathbf{A}-\mathbf{P})\sup_{\mathbf{\Xi}\in\mathcal{T}}\|\mathbf{\Xi}\|_* \\ 
& \leq 2\sqrt{r}nR^*\sigma_1(\mathbf{A}-\mathbf{P})
\end{align*}
by matrix norm inequalities $|\mathrm{tr}(\mathbf{B^\top C})|\leq \|\mathbf{B}\|_2\|\mathbf{C}\|_*$ and $\|\mathbf{C}\|_*\leq \sqrt{r}\|\mathbf{C}\|_F\leq \sqrt{r}n\|\mathbf{C}\|_{\max}$ for $\mathrm{rank}\mathbf{C}\leq r$.
Together with Markov's inequality and the fact that 
\begin{align*}
\mathbb{E}[\sigma_1(\mathbf{A}-\mathbf{P})] &\leq C_0\bigg(\Big(\max_i\sum_j\mathbb{E}[A_{ij}^2]\Big)^{\frac{1}{2}}+\Big(\max_j\sum_i\mathbb{E}[A_{ij}^2]\Big)^{\frac{1}{2}}+\sum_{ij}\mathbb{E}[A_{ij}^4]\Big)^{\frac{1}{4}}\bigg) \\
&\leq C_0\sqrt{n}(2\sqrt{M_2}+\sqrt[4]{M_4})
\end{align*}
 by Latala's theorem \citep{Lataa2005} where $C_0$ is some universal constant, we have 
\begin{align}
\mathbb{P}\Big(2\sup_{\mathbf{\Xi}\in\mathcal{T}}|\mathrm{tr}((\mathbf{A}-\mathbf{P})^\top\mathbf{\Xi})|\geq n^2\delta\Big) &\leq \mathbb{P}(2\sqrt{r}nR^*\sigma_1(\mathbf{A}-\mathbf{P})\geq n^2t) \notag\\
&\leq \frac{2\sqrt{r}R\mathbb{E}[\sigma_1(\mathbf{A}-\mathbf{P})]}{nt}\notag \\
&\leq \frac{2\sqrt{r}R C_0(2\sqrt{M_2}+\sqrt[4]{M_4})}{\sqrt{n}t}. \label{ineq1term1}
\end{align}

For the second term in (\ref{ineq1}), 
\begin{align}
&\quad \mathbb{P}\bigg(\Big|\sum^m_{k=1}(\widehat{\beta}_k-\beta_k)\mathrm{tr}((\mathbf{A}-\mathbf{P})^\top\mathbf{X}_k)\Big|\geq n^2t\bigg) \notag\\
&\leq \mathbb{P}\bigg(2\sup_{\|\boldsymbol{\beta}\|_{\max}\leq K_\beta}\|\boldsymbol{\beta}\|_{\max}\Big|\sum^m_{k=1}\mathrm{tr}((\mathbf{A}-\mathbf{P})^\top\mathbf{X}_k)\Big|\geq n^2t\bigg) \notag\\
&\leq \frac{4K_{\beta}^2\operatorname{Var}\big(\sum^m_{k=1}\mathrm{tr}(\mathbf{A}^\top\mathbf{X}_k)\big)}{n^4t^2} \notag\\
&\leq \frac{4K_{\beta}^2K_x^2M_2}{n^2t^2}. \label{ineq1term2}
\end{align}
Thus, the desired result follows from (\ref{ineq1term1}), (\ref{ineq1term2}), and Lemma \ref{fro_tv_ineq}.
\end{proof}

\begin{proof}[Proof of Corollary \ref{cor1}]
The result is obtained by replacing (\ref{ineq1term1}) with Talagrand's inequality
\begin{align*}
&\quad \mathbb{P}(2\sqrt{r}nR^*\sigma_1(\mathbf{A}-\mathbf{P})\geq n^2t) \\
&\leq \mathbb{P}\Big(|\sigma_1(\mathbf{A}-\mathbf{P})-\mathbb{E}[\sigma_1(\mathbf{A}-\mathbf{P})]|\geq \frac{nt}{2\sqrt{r}R}-C_0(2\sqrt{M_2}+\sqrt[4]{M_4})\sqrt{n}\Big)\notag \\
&\leq C_1\exp\Big(-C_2\Big(\frac{nt}{2\sqrt{r}R}-C_0(2\sqrt{M_2}+\sqrt[4]{M_4})\sqrt{n}\Big)^2_+\Big),
\end{align*}
where $C_1$ and $C_2$ are some universal constants, and (\ref{ineq1term2}) with Hoeffiding's inequality
\[\mathbb{P}\bigg(\Big|\sum^m_{k=1}(\widehat{\beta}_k-\beta_k)\mathrm{tr}((\mathbf{A}-\mathbf{P})^\top\mathbf{X}_k)\Big|\geq n^2t\bigg)\leq 2\exp\Big(-\frac{n^2t^2}{4K_\beta^2K_x^2}\Big).\]
\end{proof}

\begin{proof}[Proof of Corollary \ref{col2}]
By Taylor's expansion, for some $\eta_{ij}$ between $\widehat{p}_{ij}$ and $p_{ij}$ for $i,j=1,\dots,n$,
\begin{align*}
 \|\widehat{\mathbf{\Theta}}-\mathbf{\Theta} + \mathcal{X}\otimes(\widehat{\boldsymbol{\beta}}-\boldsymbol{\beta})\|_F&=\bigg(\sum_{ij}\big(L(\widehat{p}_{ij})-L(p_{ij})\big)^2\bigg)^{\frac{1}{2}} \\
&\leq \sup_{ij}L'(\eta_{ij})\|\widehat{\mathbf{P}}-\mathbf{P}\|_F \\
&\leq \frac{1}{\inf_{ij}b''(L(\eta_{ij}))}\|\widehat{\mathbf{P}}-\mathbf{P}\|_F \\
&\leq \frac{1}{\inf_{ij}\mathrm{Var}(A_{ij})}\|\widehat{\mathbf{P}}-\mathbf{P}\|_F.
\end{align*}
Hence, the convergence of the linear predictor $\widehat{\mathbf{\Theta}} + \mathbf{X}\otimes\widehat{\boldsymbol{\beta}}$ follows from $\inf_{ij}\mathrm{Var}(A_{ij})$ being bounded away from 0. Since

\begin{align*}
 \frac{\big|\mathrm{tr}\Big((\widehat{\mathbf{\Theta}}-\mathbf{\Theta})^\top\big(\mathcal{X}\otimes(\widehat{\boldsymbol{\beta}}-\boldsymbol{\beta})\big)\Big)\big|}{\|\widehat{\mathbf{\Theta}}-\mathbf{\Theta}\|_F\|\mathcal{X}\otimes(\widehat{\boldsymbol{\beta}}-\boldsymbol{\beta})\|_F}&\leq \frac{\sum^n_{i=1}\sigma_i(\widehat{\mathbf{\Theta}}-\mathbf{\Theta})\sigma_i\big(\mathcal{X}\otimes(\widehat{\boldsymbol{\beta}}-\boldsymbol{\beta})\big)}{\|\widehat{\mathbf{\Theta}}-\mathbf{\Theta}\|_F\|\mathbf{X}\otimes(\widehat{\boldsymbol{\beta}}^*-\boldsymbol{\beta})\|_F} \\
&= \frac{\sum^{2r}_{i=1}\sigma_i(\widehat{\mathbf{\Theta}}-\mathbf{\Theta})\sigma_i\big(\mathcal{X}\otimes(\widehat{\boldsymbol{\beta}}-\boldsymbol{\beta})\big)}{\|\widehat{\mathbf{\Theta}}-\mathbf{\Theta}\|_F\|\mathcal{X}\otimes(\widehat{\boldsymbol{\beta}}-\boldsymbol{\beta})\|_F} \\
&\leq \frac{\Big(\sum^{2r}_{i=1}\sigma_i^2\big(\mathcal{X}\otimes(\widehat{\boldsymbol{\beta}}-\boldsymbol{\beta})\big)\Big)^{\frac{1}{2}}}{\|\mathcal{X}\otimes(\widehat{\boldsymbol{\beta}}-\boldsymbol{\beta})\|_F} \\
&\leq \sqrt{\delta}
\end{align*}
by the condition on the spectral distribution of $\mathcal{X}\otimes\boldsymbol{\beta}$, we see that
\begin{align*}
&\quad \|\widehat{\mathbf{\Theta}}-\mathbf{\Theta} +\mathcal{X}\otimes(\widehat{\boldsymbol{\beta}}-\boldsymbol{\beta})\|_F^2 \\
&=\|\widehat{\mathbf{\Theta}}-\mathbf{\Theta}\|_F^2+\|\mathcal{X}\otimes(\widehat{\boldsymbol{\beta}}-\boldsymbol{\beta})\|_F^2+2\mathrm{tr}\Big((\widehat{\mathbf{\Theta}}-\mathbf{\Theta})^\top\big(\mathcal{X}\otimes(\widehat{\boldsymbol{\beta}}-\boldsymbol{\beta})\big)\Big)\\
&\geq \|\widehat{\mathbf{\Theta}}-\mathbf{\Theta}\|_F^2+\|\mathcal{X}\otimes(\widehat{\boldsymbol{\beta}}-\boldsymbol{\beta})\|_F^2-2\sqrt{\delta}\|\widehat{\mathbf{\Theta}}-\mathbf{\Theta}\|_F\|\mathcal{X}\otimes(\widehat{\boldsymbol{\beta}}-\boldsymbol{\beta})\|_F\\
&\geq (1-\sqrt{\delta})(\|\widehat{\mathbf{\Theta}}-\mathbf{\Theta}\|_F^2+\|\mathcal{X}\otimes(\widehat{\boldsymbol{\beta}}-\boldsymbol{\beta})\|_F^2).
\end{align*}

Thus, by Theorem \ref{thm1},
\[n^{-1}\|\widehat{\mathbf{\Theta}}-\mathbf{\Theta}\|_F\stackrel{p}{\longrightarrow} 0\]
and 
\[(\widehat{\boldsymbol{\beta}}-\boldsymbol{\beta})^\top\Big(n^{-2}\sum_{ij}\mathbf{x}_{ij}\mathbf{x}_{ij}^\top\Big)(\widehat{\boldsymbol{\beta}}-\boldsymbol{\beta})=n^{-2}\|\mathcal{X}\otimes(\widehat{\boldsymbol{\beta}}-\boldsymbol{\beta})\|_F^2\stackrel{p}{\longrightarrow} 0. \]
\end{proof}

\begin{proof}[Proof of Theorem \ref{thm2}]
Let $\widehat{\mathbf{\Theta}}^*=\operatorname*{\arg\min}_{\mathbf{\Xi}\in\mathcal{T}}\|\mathbf{\Xi}-\mathbf{\Theta}\|_F$.
\begin{align}
n^2D_{KL}(f_{\mathbf{P}}\|f_{\widehat{\mathbf{P}}})
&=h(\mathbf{\Theta},\boldsymbol{\beta})-h(\widehat{\mathbf{\Theta}},\widehat{\boldsymbol{\beta}}) \notag\\
&\leq \ell_{\mathbf{A},\mathbf{X}}(\widehat{\mathbf{\Theta}},\widehat{\boldsymbol{\beta}})-h(\widehat{\mathbf{\Theta}},\widehat{\boldsymbol{\beta}}) - \ell_{\mathbf{A},\mathbf{X}}(\widehat{\mathbf{\Theta}}^*,\boldsymbol{\beta}) + h(\widehat{\mathbf{\Theta}}^*,\boldsymbol{\beta}) \label{lrem:dl_diff} \\
&\quad -h(\widehat{\mathbf{\Theta}}^*,\boldsymbol{\beta})+h(\mathbf{\Theta},\boldsymbol{\beta})\notag\\
&=\mathrm{tr}((\mathbf{A}-\mathbf{P})^\top(\widehat{\mathbf{\Theta}}-\widehat{\mathbf{\Theta}}^*))+\sum^m_{k=1}(\widehat{\beta}_k-\beta_k)\mathrm{tr}((\mathbf{A}-\mathbf{P})^\top\mathbf{X}_k) \notag\\
&\quad + \mathrm{tr}(\mathbf{P}^\top(\mathbf{\Theta}-\widehat{\mathbf{\Theta}}^*)) + \sum_{ij}\big(b(\widehat{\theta}_{ij}^*+\mathbf{x}_{ij}^\top\boldsymbol{\beta}) - b(\theta_{ij}+\mathbf{x}_{ij}^\top\boldsymbol{\beta})\big) \label{lrem:misspecification}
\end{align}
The first two terms above converge to 0 in probability by a similar argument in the proof of Theorem \ref{thm1}.
Note that 
\[\mathrm{tr}(\mathbf{P}^\top(\mathbf{\Theta}-\widehat{\mathbf{\Theta}}^*))\leq\sigma_1(\mathbf{P})\|\mathbf{\Theta}-\widehat{\mathbf{\Theta}}^*\|_*\leq n\sum^n_{k=r+1}\sigma_k(\mathbf{\Theta})\]
and that, by Taylor's expansion, for some $\xi_{ij}$ between $\widehat{\theta}_{ij}^*+\mathbf{x}_{ij}^\top\boldsymbol{\beta}$ and $\theta_{ij}+\mathbf{x}_{ij}^\top\boldsymbol{\beta}$,
\begin{align*}
\sum_{ij}\big(b(\widehat{\theta}_{ij}^*+\mathbf{x}_{ij}^\top\boldsymbol{\beta}) - b(\theta_{ij}+\mathbf{x}_{ij}^\top\boldsymbol{\beta})\big)&= \sum_{ij} b'(\xi_{ij})(\widehat{\theta}_{ij}^*-\theta_{ij}) \\
&\leq K_p\sum_{ij}|\widehat{\theta}_{ij}^*-\theta_{ij}| \\
&\leq nK_p\|\widehat{\mathbf{\Theta}}^*-\mathbf{\Theta}\|_F \\
&\leq nK_p\|\widehat{\mathbf{\Theta}}^*-\mathbf{\Theta}\|_* \\
&= nK_p\sum^n_{k=r+1}\sigma_k(\mathbf{\Theta}),
\end{align*}
Therefore,
\[D_{KL}(f_{\mathbf{P}}\|f_{\widehat{\mathbf{P}}})=O_p\bigg(n^{-1}\sum^n_{k=r+1}\sigma_k(\mathbf{\Theta})\bigg)\]
and by Lemma \ref{fro_tv_ineq}, for $\delta>0$,
\[n^{-1}\|\widehat{\mathbf{P}}-\mathbf{P}\|_F=O_p\bigg(M_{1+\delta}^{\frac{1}{1+\delta}}\Big(n^{-1}\sum^n_{k=r+1}\sigma_k(\mathbf{\Theta})\Big)^{\frac{\delta}{2+2\delta}}\bigg).\]
\end{proof}

\end{document}